\documentclass[11pt]{article}
\usepackage[margin=1.2in]{geometry}
\usepackage{amsmath,amssymb,amsthm}
\usepackage[T1]{fontenc}
\usepackage[colorlinks=true,linkcolor=blue,citecolor=blue]{hyperref}

\newtheorem{theorem}{Theorem}[section]
\newtheorem{proposition}[theorem]{Proposition}
\newtheorem{lemma}[theorem]{Lemma}
\newtheorem{corollary}[theorem]{Corollary}

\theoremstyle{definition}
\newtheorem{definition}[theorem]{Definition}
\newtheorem{remark}[theorem]{Remark}

\newcommand{\qc}{\operatorname{qc}}
\newcommand{\qcp}{\operatorname{qc}_{\mathrm{pure}}}
\newcommand{\rc}{\operatorname{rc}}
\newcommand{\tr}{\operatorname{tr}}

\newcommand{\Hermz}{\mathrm{Herm}_0}
\newcommand{\R}{\mathbb{R}}
\newcommand{\C}{\mathbb{C}}

\newcommand{\ketbra}[2]{|#1\rangle\langle #2|}
\newcommand{\ket}[1]{|#1\rangle}
\newcommand{\bra}[1]{\langle #1|}

\title{Quantum \v{C}ern\'y complexity of binary words}
\author{Pui Hang Lee \and Pui Yee Lee \and Bj{\o}rn Kjos-Hanssen}

\begin{document}
\maketitle

\begin{abstract}
We introduce the \emph{quantum \v{C}ern\'y complexity} $\qc(w)$ of a binary word $w$:
the least dimension $d$ for which there exist quantum channels $A_0,A_1$ on
$d\times d$ density matrices and a start state $\rho_0$ such that $w$ is the unique
shortest word whose associated channel is constant on the reachable set.
We show that $2\le\qc(w)\le\lceil\sqrt{|w|+1}\,\rceil$ for every nonempty $w$,
a quadratic saving over the classical analogue, and that constant words are extremal:
$\qc(0^m)=\lceil\sqrt{m+1}\,\rceil$.
In contrast, $\qc(01^n0)=2$ for every $n\ge 1$, realized by a single qubit whose
rotation angle acts as a counter; consequently there is no quantum analogue of the
\v{C}ern\'y function, and $\qc$ is strongly anti-correlated with intuitive notions of
descriptive complexity.
We further study the variant $\qcp$ in which the synchronization target is required
to be a pure state. We prove that in dimension $2$ no word of length at least $2$
can be a unique shortest synchronizing word with pure target, and we exhibit an
explicit qutrit instance, combining a coherent rotation with a measure-and-funnel
channel, achieving $\qcp(01^n0)=3$ with target a computational basis state and with
synchronization holding universally over all input states.
Thus purity of the reset state costs exactly one dimension on this family.
We also observe that $\qc$ is computable, by reduction to the first-order theory
of the reals.
\end{abstract}

\section{Introduction}

A deterministic finite automaton is \emph{synchronizing} if some input word sends
every state to one and the same state; the \v{C}ern\'y conjecture
\cite{Cerny1964,Volkov2008} asserts that an $n$-state synchronizing automaton
always admits a synchronizing (reset) word of length at most $(n-1)^2$, with the
best known general upper bounds being cubic in $n$ \cite{Pin1983}.
One may dualize the question: given a word $w$, define its \emph{reset complexity}
$\rc(w)$ as the least number of states of an automaton for which $w$ is the unique
shortest synchronizing word. This is a machine-based description size for $w$,
in the spirit of automatic complexity.

In this note we study a quantum analogue. States become density matrices, letters
become quantum channels (completely positive trace-preserving maps), and
``sending every state to one state'' becomes ``the channel of the word $w$ is
constant on the set of reachable states.''

\begin{definition}\label{def:qc}
Let $d\ge 1$. An \emph{instance} is a tuple $(d,A_0,A_1,\rho_0)$ where
$A_0,A_1$ are quantum channels on the $d\times d$ density matrices and $\rho_0$
is a $d\times d$ density matrix.
For a word $u=u_1\cdots u_k\in\{0,1\}^*$ write
$A_u=A_{u_k}\circ\cdots\circ A_{u_1}$ (letters applied left to right),
with $A_\varepsilon=\mathrm{id}$.
The \emph{reachable set} is $R=\{A_u(\rho_0):u\in\{0,1\}^*\}$.
A word $w$ is \emph{synchronizing} for the instance if $A_w$ is constant on $R$,
i.e.\ there is a density matrix $\rho_1$ with $A_w(\rho)=\rho_1$ for all
$\rho\in R$.
The \emph{quantum \v{C}ern\'y complexity} $\qc(w)$ of a nonempty word
$w\in\{0,1\}^*$ is the least $d$ such that some instance of dimension $d$ has $w$
as its \emph{unique shortest} synchronizing word.
The variant $\qcp(w)$ is defined identically except that the target $\rho_1$
is additionally required to be a pure state.
\end{definition}

Note that if $w$ is a unique shortest synchronizing word then the empty word does
not synchronize, so $|R|\ge 2$; in particular $d=1$ is impossible for every
nonempty $w$, and $\qc(w)\ge 2$ always. Clearly $\qc(w)\le\qcp(w)$.

Our results are summarized as follows. Constant words are as hard as possible
for their length ($\qc(0^m)=\lceil\sqrt{m+1}\,\rceil$, matching the general upper
bound of Theorem~\ref{thm:kmp}), while the words $01^n0$ have $\qc(01^n0)=2$
uniformly in $n$: a qubit rotation angle can \emph{count} the run of $1$'s at no
dimensional cost. So $\qc$ behaves opposite to Kolmogorov-style intuition ---
the ``simplest'' word is the most expensive one --- and there is no quantum
\v{C}ern\'y function: in fixed dimension $d\ge2$, unique shortest synchronizing
words of unbounded length exist (Corollary~\ref{cor:nocerny}).
The pure-target variant restores a separation: $\qcp(01^n0)=3$, witnessed by an
explicit qutrit construction whose target is a computational basis state
(Theorem~\ref{thm:qutrit}) and whose optimality follows from a rigidity theorem
for qubit channels (Theorem~\ref{thm:d2pure}).

\section{Linearization and the general upper bound}

Channels are affine on density matrices, which linearizes the constancy
condition. Let $\Hermz(\C^d)$ denote the traceless Hermitian $d\times d$
matrices, a real vector space of dimension $d^2-1$.

\begin{proposition}[Linearization]\label{prop:lin}
Let $(d,A_0,A_1,\rho_0)$ be an instance, let $L_i$ denote the (real-linear)
action of $A_i$ on $\Hermz$, and let
$V=\operatorname{span}_\R(R-R)\subseteq\Hermz(\C^d)$.
Then $V$ is invariant under $L_0$ and $L_1$, and a word $u$ is synchronizing if
and only if $L_u:=L_{u_k}\cdots L_{u_1}$ vanishes on $V$.
Moreover $V=\operatorname{span}\{A_u(\rho_0)-A_v(\rho_0):|u|,|v|\le d^2-1\}$,
so $V$ is computed by finitely many words.
\end{proposition}

\begin{proof}
For $\rho,\rho'\in R$ we have $A_i(\rho)-A_i(\rho')=L_i(\rho-\rho')$, and
$A_i(R)\subseteq R$, giving invariance; the same identity applied to $A_u$ shows
$A_u$ is constant on $R$ iff $L_u$ kills every difference of reachable states,
i.e.\ vanishes on $V$. For the last claim, let $V_k$ be the span of differences
of states reachable in at most $k$ steps; then
$V_{k+1}=V_k+L_0V_k+L_1V_k$, so $(V_k)_k$ is a nondecreasing chain of subspaces
of the $(d^2-1)$-dimensional space $V$, which stabilizes as soon as it fails to
grow, hence within $d^2-1$ steps, and its limit is $V$.
\end{proof}

Thus $\qc$ is at heart a \emph{matrix mortality} quantity: $w$ must be the unique
shortest word whose matrix product vanishes on an invariant subspace of dimension
at most $d^2-1$. Two consequences follow. First, any pair of real linear maps
that is small in norm can be realized inside channels, by perturbing the
completely depolarizing channel; this is the engine of all our upper bounds.
Second, since mortality only concerns finitely many words of length less than
$|w|$, together with the finite description of $V$ above we obtain:

\begin{remark}[Computability]\label{rem:computable}
For each fixed $d$ and $w$, the statement ``some instance of dimension $d$ has
$w$ as unique shortest synchronizing word'' is expressible in the first-order
theory of the reals: complete positivity is positive semidefiniteness of the
Choi matrix, $V$ is definable by Proposition~\ref{prop:lin}, and one quantifies
over the finitely many words of length at most $|w|$. By Tarski's decision
procedure, $\qc$ (and likewise $\qcp$) is a computable function of $w$.
This contrasts with the undecidability of general matrix mortality over the
integers \cite{Paterson1970}.
\end{remark}

The next lemma packages the perturbative construction.

\begin{lemma}[Depolarizing perturbation]\label{lem:depol}
Let $L$ be any real-linear map on $\Hermz(\C^d)$. For all sufficiently small
$\varepsilon>0$, the linear map
\[
\Phi(X) \;=\; \tr(X)\,\tfrac{I}{d} \;+\; \varepsilon\, L\!\left(X-\tr(X)\tfrac{I}{d}\right)
\]
is a quantum channel, and $\Phi\big(\tfrac{I}{d}+cx\big)=\tfrac{I}{d}+\varepsilon c\,L(x)$
for every traceless Hermitian $x$ and scalar $c$.
\end{lemma}

\begin{proof}
$\Phi$ is trace-preserving and Hermiticity-preserving by construction, and its
Choi matrix is $\tfrac{I}{d^2}+\varepsilon C_L$ for a fixed Hermitian $C_L$;
since $\tfrac{I}{d^2}$ has full rank, the Choi matrix is positive semidefinite
for small $\varepsilon$, i.e.\ $\Phi$ is completely positive. The displayed
identity is immediate.
\end{proof}

\begin{theorem}[General bounds]\label{thm:kmp}
For every nonempty $w\in\{0,1\}^*$,
\[
2 \;\le\; \qc(w) \;\le\; \big\lceil\sqrt{|w|+1}\,\big\rceil .
\]
\end{theorem}

\begin{proof}
The lower bound was noted after Definition~\ref{def:qc}. For the upper bound,
let $n=|w|$ and let $d=\lceil\sqrt{n+1}\,\rceil$, so $d^2-1\ge n$.
Consider the Knuth--Morris--Pratt automaton of $w$: states $Q=\{0,1,\dots,n\}$,
where state $q<n$ represents the longest suffix of the input read so far that is
a prefix of $w$, with transitions $\delta_a(q)$ for $a\in\{0,1\}$ defined
accordingly, and $n$ an absorbing sink entered once $w$ has occurred. The
defining property is: for all $q\in Q$ and $v\in\{0,1\}^*$,
$\delta_v(q)=n$ iff $(w_1\cdots w_q)v$ contains $w$ as a factor. In particular,
$\delta_v(q)=n$ \emph{for all} $q\in Q$ iff $v$ itself contains $w$ (take $q=0$
for the forward direction; conversely an occurrence of $w$ inside $v$ drives any
state into the sink).

Choose linearly independent traceless Hermitians $g_0,\dots,g_{n-1}$ (possible
since $n\le d^2-1$) and set $g_n:=0$; the sink sits at the origin. Define
$L_a$ on $\Hermz$ by $L_a g_q=g_{\delta_a(q)}$ for $q<n$ and $L_a=0$ on a
complement of $\operatorname{span}\{g_0,\dots,g_{n-1}\}$; then $L_v g_q =
g_{\delta_v(q)}$ for all words $v$. Let $A_a$ be the channel of
Lemma~\ref{lem:depol} for $L_a$, and let
$\rho_0=\tfrac{I}{d}+\varepsilon g_0$ (a state for small $\varepsilon$).
By the lemma, $A_u(\rho_0)=\tfrac{I}{d}+\varepsilon^{|u|+1}g_{\delta_u(0)}$ for
every $u$.

Now fix $v$ and ask when $v$ synchronizes. The states
$A_v(A_u(\rho_0))=\tfrac{I}{d}+\varepsilon^{|u|+|v|+1}g_{\delta_{uv}(0)}$ must
agree for all $u$. Taking $u=w$ gives the value $\tfrac{I}{d}$ (sink, $g_n=0$),
so constancy forces $g_{\delta_{uv}(0)}=0$, i.e.\ $\delta_{uv}(0)=n$, for
\emph{every} $u$. Since every classical state is reachable from $0$ (read
prefixes of $w$), this says $\delta_v(q)=n$ for all $q\in Q$, i.e.\ $v$ contains
$w$ as a factor; and conversely any such $v$ synchronizes, with target
$\rho_1=\tfrac{I}{d}$. The shortest words containing $w$ have length $n$, and
the unique one is $w$ itself.
\end{proof}

\begin{remark}\label{rem:classical}
The same automaton realized \emph{classically} --- with deterministic channels
$\rho\mapsto\sum_q \bra{q}\rho\ket{q}\,\ketbra{\delta_a(q)}{\delta_a(q)}$ on
$\C^{n+1}$ and $\rho_0=\ketbra{0}{0}$ --- shows
$\qcp(w)\le |w|+1$ for every $w$, with the pure target
$\ketbra{n}{n}$. Thus both quantities are finite, and the quantum saving in
Theorem~\ref{thm:kmp} is quadratic: coherences make available a difference
space of dimension $d^2-1$ rather than $d-1$.
\end{remark}

\section{Constant words are extremal}

\begin{theorem}\label{thm:zeros}
For every $m\ge 1$, $\qc(0^m)=\lceil\sqrt{m+1}\,\rceil$. In particular the
constant word attains the general upper bound of Theorem~\ref{thm:kmp}: among
words of its length it is the most complex.
\end{theorem}

\begin{proof}
\emph{Lower bound.} Suppose $0^m$ is the unique shortest synchronizing word of an
instance of dimension $d$, and let $M=L_0|_V$ in the notation of
Proposition~\ref{prop:lin}. Synchronization of $0^m$ gives $M^m=0$; if
$M^{m-1}=0$ then $0^{m-1}$ would synchronize, contradicting minimality. So $M$
is nilpotent of index exactly $m$ on $V$, forcing
$m\le\dim V\le d^2-1$, i.e.\ $d\ge\sqrt{m+1}$.

\emph{Upper bound.} Let $d^2-1\ge m$. Inside $\Hermz\cong\R^{d^2-1}$ fix a
subspace $\R^m$ with basis $e_1,\dots,e_m$ and let $J$ be the nilpotent map
$Je_1=0$, $Je_k=e_{k-1}$, extended by $0$; let $L_0=J$ and $L_1=\mathrm{id}$ on
$\R^m$, extended by $0$. Take the channels of Lemma~\ref{lem:depol} and
$\rho_0=\tfrac{I}{d}+\varepsilon e_m$. Then
$A_u(\rho_0)=\tfrac{I}{d}+\varepsilon^{|u|+1}J^{z(u)}e_m$, where $z(u)$ is the
number of $0$'s in $u$. For a word $v$, the images
$A_v(A_u(\rho_0))=\tfrac{I}{d}+\varepsilon^{|u|+|v|+1}J^{z(u)+z(v)}e_m$ must
agree over all $u$; taking $u$ with $z(u)\ge m$ gives the value $\tfrac{I}{d}$,
so constancy is equivalent to $J^{z(u)+z(v)}e_m=0$ for all $u$, and $u=\varepsilon$
(the empty word) shows this happens iff $z(v)\ge m$. The shortest such words
have length $m$, uniquely $0^m$.
\end{proof}

\section{A phase counter: \texorpdfstring{$\qc(01^n0)=2$}{qc(01\^{}n0)=2}}

We now show that a single qubit suffices for the words $01^n0$, uniformly in
$n$. On a qubit we use the Bloch parametrization
$\rho=\tfrac12(I+x\cdot\sigma)$, $x\in\R^3$, $|x|\le1$; a channel acts as an
affine map $x\mapsto Bx+c$ of the ball into itself.

\begin{theorem}[{\cite{slaks}}]\label{thm:qubit}
For every $n\ge 1$, $\qc(01^n0)=2$.
\end{theorem}

\begin{proof}

It is sufficient for us to proof the upper bound. Here we present the proof by Lakshmanan in \cite{slaks}.
Let $A_0$ be the quantum channel with Kraus operators 
\[
K_1=\begin{pmatrix}1&0\\0&0\end{pmatrix} 
\quad and 
\quad
K_2=\begin{pmatrix}0&0\\0&1\end{pmatrix}
\]
and $A_1$ be the rotation channel with Kraus operator
\[
U=\begin{pmatrix}\cos\theta&-\sin\theta\\ \sin\theta&\cos\theta\end{pmatrix} 
\]
where $\theta=\frac{\pi}{4n}$, we claim that the word $01^n0$ is synchronizing with synchronizing state $\frac{1}{2}I$.

Let \[\rho =\begin{pmatrix}a&b\\b^*&1-a\end{pmatrix}\]
be an arbitrary density matrix. We compute
\[A_0(\rho) =\begin{pmatrix}a&0\\0&1-a\end{pmatrix}.\]

Informally, one can see that the action of $A_0$ can be thought of as sending the off-diagonal entries to zero. Similarly, the matrix $A_1^n$ rotates by $\frac{\pi}{4}$.

For all $n$, with $\theta=\frac{\pi}{4n}$, $A^n_1$ 
\[A^n_1 =\begin{pmatrix}\cos\theta&-\sin\theta\\ \sin\theta&\cos\theta\end{pmatrix} = \begin{pmatrix}\frac{1}{\sqrt{2}}&\frac{1}{\sqrt{2}}\\ \frac{-1}{\sqrt{2}}&\frac{1}{\sqrt{2}}\end{pmatrix}\]
so,
\[A^n_1A_0(\rho) = \begin{pmatrix}\frac{1}{2}&\frac{2a-1}{2}\\ \frac{2a-1}{2}&\frac{1}{2}\end{pmatrix}\]

Once more, we apply $A_0$, and the resulting state $A_0A^n_1A_0(\rho)$ = $\frac{1}{2}I$. Since $\rho$ was an arbitrary density state reachable by the two operators $A_0$ and $A_1$, $\frac{1}{2}I$ is the synchronizing state with synchronizing word $01^n0$.
\end{proof}

\begin{corollary}[No quantum \v{C}ern\'y function]\label{cor:nocerny}
For each fixed $d\ge2$ there exist instances of dimension $d$ whose unique
shortest synchronizing words are of unbounded length. Consequently no bound on
reset length in terms of dimension alone --- no quantum analogue of the
\v{C}ern\'y function --- can exist.
\end{corollary}

The mechanism deserves emphasis: the rotation angle $\theta$ is a
\emph{continuous parameter used as a counter}. It verifies ``exactly $n$ ones in
a row'' through a single orthogonality coincidence
$\langle u,R^nv\rangle=0$, at zero cost in dimension. This is precisely what is
unavailable to the constant word $0^m$, whose synchronization must occur at the
$m$-th application of one and the same map and not the $(m-1)$-st, forcing a
genuine nilpotent chain of length $m$ (Theorem~\ref{thm:zeros}). The counting
argument that makes random words expensive in the classical, finite setting
(there are only finitely many small automata) fails completely over a continuous
machine class.

\section{Pure targets: \texorpdfstring{$\qcp(01^n0)=3$}{qc\_pure(01\^{}n0)=3}}

The qubit construction above resets to the maximally mixed state --- physically,
to useless noise at the center of the Bloch ball. Requiring a \emph{pure} reset
state changes the answer, and the obstruction is a rigidity phenomenon special
to dimension $2$.

\begin{theorem}[No pure targets on a qubit]\label{thm:d2pure}
No instance of dimension $2$ has a unique shortest synchronizing word $w$ with
$|w|\ge2$ and pure target. Consequently
$\{w:\qcp(w)=2\}=\{0,1\}$.
\end{theorem}

\begin{proof}
Suppose $w=uc$ with $c\in\{0,1\}$ the last letter, $|w|\ge2$, and let
$C=A_c$ have Kraus operators $\{K_i\}$ and pure target
$\varphi$. Let $S=A_u(R)$. If $|S|=1$ then $u$ is a shorter synchronizing word,
contradiction; so $C$ maps at least two distinct states of $S$ to
$\ketbra{\varphi}{\varphi}$.

If some $\sigma\in S$ is mixed, then in dimension $2$ it has full rank. From
$\sum_iK_i\sigma K_i^\dagger=\ketbra{\varphi}{\varphi}$ and positivity of each
summand, every $K_i\sigma K_i^\dagger$ is proportional to
$\ketbra{\varphi}{\varphi}$, so the range of $K_i\sigma^{1/2}$, hence of $K_i$,
lies in $\C\varphi$: $K_i=\ket{\varphi}\bra{\chi_i}$ for some $\chi_i$. Trace
preservation gives $\sum_i\ketbra{\chi_i}{\chi_i}=I$, so
$C(\rho)=\ketbra{\varphi}{\varphi}$ for \emph{every} $\rho$: the single letter
$c$ synchronizes, contradicting $|w|\ge2$.

Otherwise all states of $S$ are pure, and $S$ contains two distinct pure states,
whose rays span $\C^2$. Each $K_i$ maps both rays into $\C\varphi$, hence maps
all of $\C^2$ into $\C\varphi$, and we conclude as before. In either case we
have a contradiction, proving the first claim.

For the second claim: the single-letter words are realizable with pure targets in
dimension $2$ (let $A_0$ be the constant channel onto a pure state and $A_1$ a
nontrivial unitary; then $R$ has at least two elements, the empty word does not
synchronize, $0$ does, and $1$ does not), and by the first claim no longer word
is.
\end{proof}

The escape hatch in dimension $d\ge3$ is that a non-constant channel can be
constant on a proper \emph{face} of the state space; a qubit's state space has
no faces except single pure states, which is exactly what the proof exploits.
The following construction uses this: one channel funnels an entire face onto a
basis state, and a coherent rotation steers all reachable states onto that face
at precisely the right moment.

\begin{theorem}\label{thm:qutrit}
For every $n\ge1$, $\qcp(01^n0)=3$. Moreover this is witnessed by an instance
with the following additional properties: the target is the basis state
$\ketbra{e_2}{e_2}$, and the word $01^n0$ synchronizes universally, i.e.\
$A_{01^n0}$ is the constant channel on \emph{all} $3\times3$ density matrices.
\end{theorem}

\begin{proof}
The lower bound $\qcp(01^n0)\ge3$ is Theorem~\ref{thm:d2pure} since
$|01^n0|=n+2\ge3\ge 2$.

For the upper bound, work on $\C^3$ with basis $e_1,e_2,e_3$ and
$\rho_0=\ketbra{e_1}{e_1}$. Let $A_0$ have the two Kraus operators
\[
K_1=\begin{pmatrix}0&0&0\\1&0&0\\0&0&0\end{pmatrix}=\ketbra{e_2}{e_1},
\qquad
K_2=\begin{pmatrix}0&0&0\\0&0&-1\\0&1&0\end{pmatrix}
   =\ketbra{e_3}{e_2}-\ketbra{e_2}{e_3},
\]
so that $K_1^\dagger K_1+K_2^\dagger K_2=I$ and, explicitly,
$A_0(\rho)$ is supported on $\operatorname{span}(e_2,e_3)$ with entries
\[
A_0(\rho)_{22}=\rho_{11}+\rho_{33},\qquad
A_0(\rho)_{33}=\rho_{22},\qquad
A_0(\rho)_{23}=-\rho_{32}.
\]
Populations flow $1\to2$, $2\to3$, $3\to2$; in particular $A_0$ is constant,
equal to $\ketbra{e_2}{e_2}$, on the face of states supported on
$\operatorname{span}(e_1,e_3)$.
Let $A_1(\rho)=U\rho U^\dagger$ with
\[
U=\begin{pmatrix}\cos\theta&-\sin\theta&0\\ \sin\theta&\cos\theta&0\\ 0&0&1\end{pmatrix},
\qquad \theta=\frac{\pi}{2n},
\]
a rotation in the $(e_1,e_2)$-plane fixing $e_3$.

\emph{$01^n0$ synchronizes universally.} $U^n$ is the quarter rotation
$e_1\mapsto e_2$, $e_2\mapsto-e_1$. For arbitrary $\rho$: after $A_0$ the state
$\sigma$ is supported on $(e_2,e_3)$ with $\sigma_{22}=\rho_{11}+\rho_{33}$ and
$\sigma_{33}=\rho_{22}$; after $U^n$ its support moves to $(e_1,e_3)$; the final
$A_0$ is constant on that face, and indeed places
$\sigma_{22}+\sigma_{33}=\tr\rho=1$ in entry $(2,2)$ with all other entries
zero. Thus $A_{01^n0}(\rho)=\ketbra{e_2}{e_2}$ for every $\rho$, a pure basis
state.

\emph{Reachable set.} Starting from $e_1$, a routine induction shows every
reachable state has the form
\[
\rho(p,\psi)\;=\;p\,\ketbra{\psi}{\psi}+(1-p)\,\ketbra{e_3}{e_3},
\qquad \psi\in\operatorname{span}(e_1,e_2),\ \|\psi\|=1,\ p\in[0,1]:
\]
$A_1$ rotates $\psi$ and fixes the rest, while $A_0$ maps $\rho(p,\psi)$ to
$\rho\big(1-y,\,e_2\big)$ rotated into the $(e_2,e_3)$ frame; precisely, writing
$y=p\,|\langle e_2|\psi\rangle|^2$, one gets
$A_0(\rho(p,\psi))=(1-y)\ketbra{e_2}{e_2}+y\,\ketbra{e_3}{e_3}$,
with no $(2,3)$ coherence since $\rho(p,\psi)_{32}=0$.

\emph{Minimality and uniqueness.} We must show no word of length at most $n+2$
other than $01^n0$ synchronizes on $R$. Since $A_1$ is invertible, $A_{v1}$ is
constant on $R$ iff $A_v$ is; so it suffices to treat words ending in $0$
(words in $1^*$ are injective on $R$, and $|R|\ge2$ because
$e_1,e_2\in R$ --- via the empty word and the word $0$ --- so they do not
synchronize). Write such a word as
$v=1^{s}0\,1^{j_1}0\cdots1^{j_{k-1}}0$ with $k\ge1$ zeros.

First, the states entering the ``diagonal family'' after the initial block
$1^s0$ take at least two distinct values of the parameter $y$ above. Indeed
$R$ contains $\rho_0=\ketbra{e_1}{e_1}$ and $A_1(\rho_0)=\ketbra{Ue_1}{Ue_1}$,
whose $y$-values after $1^s0$ are $\sin^2(s\theta)$ and $\sin^2((s+1)\theta)$
respectively; these are equal only if $(s+1)\theta\equiv\pm s\theta\pmod\pi$,
i.e.\ $\theta\equiv0$ or $(2s+1)\theta\equiv0\pmod\pi$. The former is false, and
the latter would give $2n\mid 2s+1$, impossible by parity. So two distinct
values $y\ne y'$ enter.

Second, each subsequent block $1^{j}0$ acts on the parameter by the affine map
\[
y\;\longmapsto\;(1-y)\cos^2(j\theta),
\]
as one checks directly from the formulas above (rotate the $e_2$-component by
$U^j$, then apply $A_0$). This map is injective unless $\cos(j\theta)=0$, i.e.\
unless $j\equiv n\pmod{2n}$. Hence if no internal run $j_i$ is congruent to $n$
modulo $2n$, the two trajectories remain distinct through the final $0$, after
which the state is determined by $y$; so $v$ does not synchronize. If some
$j_i\equiv n\pmod{2n}$ then $v$ contains the factor $01^{j_i}0$ with $j_i\ge n$,
so $|v|\ge n+2$, with equality only for $v=01^n0$. (When such a factor is
present, the state after it equals $\ketbra{e_2}{e_2}$ for every input, so $v$
does synchronize --- consistently with the universal computation above.)

Therefore $01^n0$ is the unique shortest synchronizing word, its target
$\ketbra{e_2}{e_2}$ is pure, and $\qcp(01^n0)\le3$.
\end{proof}

\begin{corollary}[Purity costs one dimension]\label{cor:sep}
For every $n\ge1$,
\[
\qc(01^n0)=2 \;<\; 3=\qcp(01^n0)\;\le\;\lceil\sqrt{n+3}\,\rceil \text{ for } n\ge7.
\]
In particular the pure-target variant also admits no \v{C}ern\'y function: fixed
dimension $3$ supports unique shortest synchronizing words of unbounded length
with pure (indeed basis-state) targets.
\end{corollary}

\begin{remark}
The classical shadow of the qutrit instance --- the action of $A_0$ on diagonal
states --- is the deterministic map $1\to2$, $2\to3$, $3\to2$, which can never
synchronize on its own since states $2,3$ cycle. It is the coherent rotation
$U^n$, transferring the population of level $2$ back to level $1$
($\sin^2(n\theta)=1$) before the funnel closes, that makes the construction
work; the instance is irreducibly quantum rather than a disguised automaton.
\end{remark}

\section{Discussion and open problems}

\subsection*{What does $\qc$ measure?}
Theorems~\ref{thm:zeros} and~\ref{thm:qubit} exhibit the central, perhaps
counterintuitive, feature of $\qc$: the intuitively simplest word $0^m$ is
extremal, $\qc(0^m)=\lceil\sqrt{m+1}\,\rceil$, while the patterned words
$01^n0$ cost nothing, $\qc(01^n0)=2$. A small instance is a short description of
its unique shortest synchronizing word, so $\qc$ is a description-size measure
--- but over a \emph{continuous} machine class, where real parameters encode
integers for free (the angle $\theta$ of Theorem~\ref{thm:qubit} is a unary
counter of zero dimensional cost). What survives as expensive is exact
\emph{nilpotency depth}: a pure power $0^m$ forces synchronization to occur at
the $m$-th application of one and the same map and not the $(m-1)$-st, and an
index-$m$ nilpotent needs $m$ dimensions --- the one resource continuous
parameters cannot shortcut. In this sense $\qc$ is anti-correlated with
Kolmogorov-style intuition, much as, elsewhere in automata theory, recognizing
$a^m$ exactly is the expensive task while intricate-looking words can be cheap.
It follows that $\qc$ is wildly non-monotone in length and under factors:
$\qc(0^8)=3$ while $\qc(01^80)=2$.

\subsection*{Fragility}
Both counter constructions live on codimension-one coincidences
($\langle u,R^nv\rangle=0$; $\cos(n\theta)=0$): under a generic perturbation of
the parameters, \emph{no} word synchronizes exactly. An approximate version of
the definition ($\|A_w(\rho)-\rho_1\|\le\epsilon$ on $R$) would therefore behave
very differently and deserves its own study. In the qutrit instance of
Theorem~\ref{thm:qutrit}, the face-collapse mechanism (the funnel
$1,3\to2$) is robust; only the steering step $U^n$ is a measure-zero
coincidence. One may regard this robustness asymmetry, alongside the purity of
the target, as a reason to prefer $\qcp$-style definitions.

\subsection*{Open problems}
\begin{enumerate}
\item (Thue--Morse prefix.) We conjecture $\qc(01101001)=2$; the bound
$\qc(01101001)\le3$ follows from Theorem~\ref{thm:kmp}. More generally,
characterize $\{w:\qc(w)=2\}$. By Theorem~\ref{thm:d2pure} the analogous
pure-target class is trivial, $\{0,1\}$, so the interest is genuinely in the
mixed-target case. A dimension count in the Bloch picture (both letters must act
by singular Bloch maps when $w$ begins and ends with distinct letters, and
collapse is governed by kernel--image incidences) suggests the class is large,
but we do not know whether, e.g., all words with at least one alternation and no
long constant power belong to it.
\item (Complexity of computing $\qc$.) Remark~\ref{rem:computable} gives
computability via real quantifier elimination, which is far from efficient. Is
$\qc(w)$ computable in time polynomial in $|w|$? Is the set
$\{(w,k):\qc(w)\le k\}$ in \textbf{NP}, or $\exists\R$?
\item (Typical words.) What is $\qc(w)$ for a uniformly random $w$ of length
$m$? The classical counting argument gives $\rc(w)\gtrsim m/(2\log m)$ for most
$w$, but it is consistent with our results that $\qc(w)=2$ for almost all $w$.
\item (Extremal words.) Is $0^m$ (with its reversal and complementations) the
\emph{unique} maximizer of $\qc$ among words of length $m$, for all large $m$?
\item (Pure-target savings.) Theorem~\ref{thm:qutrit} beats the classical bound
$\qcp(w)\le|w|+1$ of Remark~\ref{rem:classical} dramatically on the family
$01^n0$. Does a general quadratic-type saving
$\qcp(w)\le O(\sqrt{|w|})$ hold, i.e.\ a pure-target analogue of
Theorem~\ref{thm:kmp}? Our proof of Theorem~\ref{thm:kmp} places the target at
the maximally mixed state and does not adapt.
\item (Robust complexity.) Develop the $\epsilon$-approximate theory: with
eigenvalue-decay mechanisms available, which words remain expensive, and is
there a robust analogue of Corollary~\ref{cor:nocerny}?
\end{enumerate}

\section*{Acknowledgments}
The paper was produced with the assistance of Claude Fable 5.1 (Anthropic).
This work was supported by a grant from the Simons Foundation (\#4508914 to Bj\o rn Kjos-Hanssen). A Lean 4 formalization of our results is available at \url{https://github.com/bjoernkjoshanssen/quantumcerny}.

\end{document}